\documentclass[12pt]{article}
\usepackage{graphicx}
\usepackage[all]{xy}
\usepackage{color}
\usepackage{latexsym}
\usepackage{amsmath,amstext}
\usepackage{amssymb}
\usepackage{epsfig}
\usepackage{graphics}
\usepackage{euscript}
\usepackage{ifthen}
\usepackage{wrapfig}
\usepackage{amsthm}
\usepackage{multicol}
\usepackage{paralist}
\usepackage{verbatim}
\usepackage[colorlinks=true, urlcolor=blue, linkcolor=blue, citecolor=blue]{hyperref}
\usepackage[margin=1in]{geometry}
\usepackage{placeins}
\usepackage[linesnumbered,lined,commentsnumbered,ruled]{algorithm2e}
\usepackage{caption}
\usepackage{subcaption}
\usepackage{longtable}
\usepackage{lscape}
\usepackage{fancyhdr}
\usepackage{listings}
\usepackage{mathtools}
\usepackage{appendix}
\usepackage{float}
\usepackage{textcomp}
\usepackage[inline]{enumitem}

\usepackage[utf8]{inputenc}
\usepackage{csquotes}
\usepackage{xparse}
\usepackage[dvipsnames]{xcolor}

\usepackage[normalem]{ulem}

\theoremstyle{plain}
\newtheorem{theorem}{Theorem}[section]
\newtheorem{corollary}[theorem]{Corollary}
\newtheorem{lemma}[theorem]{Lemma}
\newtheorem{proposition}[theorem]{Proposition}

\theoremstyle{definition}
\newtheorem{definition}[theorem]{Definition}
\newtheorem{remark}[theorem]{Remark}
\newtheorem{example}[theorem]{Example}
\newtheorem{question}[theorem]{Question}

\numberwithin{equation}{section}

\newcommand{\R}{\mathbb{R}}

\newcommand{\C}{\mathbb{C}}

\newcommand{\G}{\mathcal{G}}

\begin{document}

\title{The Model-Specific Markov Embedding Problem for Symmetric Group-Based Models}

\author{Muhammad Ardiyansyah,
        Dimitra Kosta, and
        Kaie Kubjas}

\maketitle

\begin{abstract}
\noindent
We study model embeddability, which is a variation of the famous embedding problem in probability theory, when apart from the requirement that the Markov matrix is the matrix exponential of a rate matrix, we additionally ask that the rate matrix follows the model structure. We provide
a characterisation of model embeddable Markov
matrices corresponding to symmetric group-based phylogenetic
models. In particular, we provide necessary and
sufficient conditions in terms of the eigenvalues of symmetric group-based matrices. To showcase
our main result on model embeddability,
we provide an application to hachimoji models, which are eight-state models for synthetic DNA. Moreover, our main result on model embeddability enables us to
compute the volume of the set of model embeddable Markov matrices relative to the volume of other relevant sets of Markov matrices within the model.
\end{abstract}

\section{Introduction}
The embedding problem for stochastic matrices, also known as Markov matrices, deals with the question of deciding whether a stochastic matrix $M$ is the matrix exponential of a rate matrix $Q$. A rate matrix, also known as a Markov generator, has non-negative non-diagonal entries and row sums equal to zero. If a stochastic matrix satisfies such a property and can be expressed as a matrix exponential of a rate matrix, namely $M= e^{Qt}$, then $M$ is called embeddable. Applications of the embeddability property vary from biology and nucleotide substitution models \cite{verbyla2013embedding} to mathematical finance \cite{Israel_etal}. For a first formulation of the embedding problem see \cite{Elfving}. An account of embeddable Markov matrices is provided in \cite{Davies}. The embedding problem for $2 \times 2$ matrices is due to Kendall and first published by Kingman \cite{Kingman},  for $3 \times 3$ matrices is fully settled in a series of papers  \cite{Carette,ChenChen,Cuthbert,Israel_etal,Johansen}, while for $4 \times 4$ stochastic matrices has been recently solved in \cite{casanellas2020embedding}. In general, when the size $n$ of the stochastic matrix is greater than 4, the work \cite{casanellas2020embedding} establishes a criterion for deciding whether a generic $n\times n$ Markov matrix with distinct eigenvalues is embeddable and proposes an algorithm that lists all its Markov generators. In the present paper we study a refinement of the classical embedding problem, called the model embedding problem for a class of $n \times n$ stochastic matrices coming from phylogenetic models.

Phylogenetics is the field that aims at reconstructing the history of evolution of species.  A phylogenetic model is a mathematical model used to understand the evolutionary process given genetic data sets. The most popular phylogenetic models are nucleotide substitution models which use aligned DNA sequence data to study the molecular evolution of DNA. A comprehensive treatment of phylogenetic methods is given by Felsenstein, who is considered the initiator of statistical phylogenetics, in his seminal book \cite{Felsenstein}. Algebraic and geometric methods have been employed with great success in the study of phylogenetic models leading to an explosion of related research work and the establishment of the field of phylogenetic algebraic geometry, also known as algebraic phylogenetics; see \cite{AllmanRhodes03,banos2016phylogenetic,MartaJesus07,CavenderFelsenstein,
gross2018distinguishing,EvansSpeed,Lake,PachterSturmfels,SturmfelsSullivant} for a non-exhaustive list of publications.

To build such a phylogenetic model, we first require a phylogenetic tree $T$, which is a directed acyclic graph comprising of vertices and edges representing the evolutionary relationships of a group of species.
The vertices with valency 1 are called the leaves of the tree.

The tree is considered rooted and the direction of evolution is from the root towards the leaves. On each vertex of the tree $T$, we associate a random variable with $k$ possible states, where in phylogenetics $k$ is often taken to be 2, for the binary states $\{0,1\}$,
or 4, to represent the four types of DNA nucleotides $\{ \texttt{A, C, G, T} \}$. We also require a transition matrix (also known as a mutation matrix) $M^{(e)}$ corresponding to each edge $e$ of the tree, where the entries of this $k\times k$ matrix $M^{(e)}$ represent the probabilities of transition between states. In a phylogenetic tree, the leaves correspond to extant species and so the random variables at the leaves are observed, while the interior vertices correspond to possibly extinct species and so the random variables at the interior vertices are hidden.

In this paper we are focusing on symmetric group-based substitution models.  Substitution models are a class of phylogenetic models which make the fundamental assumption that sites evolve independently according to a continuous-time Markov process and for which the transition matrices are stochastic matrices of the form $M^{(e)}=\exp\left(t_eQ^{(e)}\right)$. Group-based models are a special class of substitution models, in which the matrices $Q^{(e)}$ can be pairwise distinct, but they can all be simultaneously diagonalizable by a linear change of coordinates given by the discrete Fourier transform of an abelian group, also called a commutative group. For example, the Cavender-Farris-Neyman (CFN) model \cite{Cavender,Farris,Neyman}, as well as the Jukes-Cantor (JC) \cite{JukesCantor},  the Kimura-2 parameter (K2P) \cite{Kimura80} and the Kimura-3 parameter (K3P) \cite{Kimura81} models for DNA are all group-based phylogenetic models. In \cite{SturmfelsSullivant}, it is established that through the discrete Fourier transform group-based models correspond to toric varieties, which are geometric objects with nice combinatorial properties. We are interested in symmetric group-based substitution models. Namely, apart from distinct and simultaneously diagonalizable, the transition matrices $Q^{(e)}$ are also symmetric square matrices.   The symmetricity assumption guarantees that the eigenvalues of rate and transition matrices of a group-based model are real, a property that we use in the proof of our main theorem.
Symmetric models are a subset of a special class of models called time-reversible models, where the Markov process appears identical when moving forward or backward in time.   Our results apply to group-based models following an ergodic time-reversible Markov process, as in this case the rate matrices $Q$ are symmetric according to \cite[Lemma 17.2]{PachterSturmfels}.

 The classical embedding problem is concerned with finding a rate matrix $Q$, given a square stochastic matrix $M$.  A variant of the embedding problem that asks for a reversible Markov generator for a stochastic matrix is studied in \cite{Jia}. When we impose the assumption that the rate matrix $Q$ follows the corresponding model conditions, we arrive at a different refined notion of embeddability called model embeddability. The embeddability of circulant and equal-input stochastic matrices is studied in \cite{Baake}.  In the current paper, we focus on the $(\mathcal{G},L)$-embeddability for $n \times n$ matrices corresponding to symmetric group-based substitution models. The $(\mathcal{G},L)$-embeddability means that we require that the rate matrices $Q$ preserve the symmetric group-based structure imposed by the abelian group $\mathcal{G}$ and the symmetric labelling $L$, which we define at the beginning of Section~\ref{section:preliminaries}. Model embeddability for symmetric group-based models is relevant both for homogeneous and inhomogeneous time-continuous processes, as group-based models are Lie Markov models, and hence multiplicatively closed~\cite{sumner2012lie,verbyla2013embedding}. A study of the set of embeddable and model-embeddable matrices corresponding to the Jukes-Cantor, Kimura-2 and Kimura-3 DNA substitution models, which are all symmetric group-based models, is undertaken in \cite{casanellas2020embeddability} and \cite{RocaFernandez}.  In particular, a full characterisation of the set of embeddable $4\times 4$ Kimura 2-parameter matrices is provided in \cite{casanellas2020embeddability}, which together with the results of \cite{RocaFernandez} fully solve the embedding problem for the Kimura 3-parameter model as well.  Although model embeddability, which is a refined notion of embeddability imposed by the model structure, implies classical embeddability, the converse is generally not true (see also \cite[Example 3.1]{RocaFernandez}).

The main result of this paper is a characterization of $(\mathcal{G},L)$-embeddability for any abelian group $\mathcal{G}$ equipped with a symmetric $\mathcal{G}$-labeling function $L$ in~Theorem \ref{theorem:embeddable_matrices}. We provide necessary and sufficient conditions which the eigenvalues of the stochastic matrix of the model need to satisfy for the matrix to be $(\mathcal{G},L)$-embeddable.  To showcase our result, we first introduce three group-based models with the underlying group $\mathbb{Z}_2 \times \mathbb{Z}_2 \times \mathbb{Z}_2$, based on the hachimoji DNA system introduced in~\cite{Hoshika}. Hachimoji, a Japanese word meaning ``eight letters'', is used to describe a
synthetic analog of the nucleic acid DNA, where we have
the four natural nucleobases $\{\texttt{A,C,G,T}\}$ and furthermore an additional
four synthetic nucleotides $\{\texttt{P,Z,B,S}\}$. We then apply Theorem~\ref{theorem:embeddable_matrices} to characterise the model embeddability for the three hachimoji DNA models. The three models are called hachimoji 7-parameter, hachimoji 3-parameter and hachimoji 1-parameter models, which can be thought of as generalisations of the Kimura 3, Kimura 2 and Jukes-Cantor models respectively. Finally, the characterisation of model embeddability in terms of eigenvalues enables us to compute the volume of the
$(\mathcal{G},L)$-embeddable Markov matrices and compare this volume with volumes of other relevant sets of Markov matrices. For the general Jukes-Cantor model, which includes the hachimoji 1-parameter model, the volumes can be derived exactly; for the hachimoji 3-parameter model and for the hachimoji 7-parameter model symbolically and numerically.

The outline of the paper is the following. Section~\ref{section:preliminaries} gives a mathematical background covering notions such as the labeling functions, group-based models and the discrete Fourier transform. Section~\ref{Section:labelling} introduces symmetric $\mathcal{G}$-compatible labelings which is a class of labeling functions with particularly nice properties. Section~\ref{Section:embeddability} presents the main result of this paper about the model embedding problem for symmetric group-based models equipped with a certain labeling function. Then in Section~\ref{section:Hachimoji} we focus on the hachimoji DNA and provide exact characterisation of the model embeddability in terms of eigenvalues of the Markov matrix for the hachimoji 7-parameter, the hachimoji 3-parameter and the hachimoji 1-parameter models. Finally, Section~\ref{section:volume} presents results on the volume of stochastic matrices that are $(\mathcal{G},L)$-embeddable for the three hachimoji group-based models. The code for the computations in this paper is available at
  \url{https://github.com/ardiyam1/Model-Embeddability-for-Symmetric-Group-Based-Models}.

\section{Preliminaries}\label{section:preliminaries}

In this section, we give background on group-based models and the discrete Fourier transform.

\begin{definition}
   Let $\G$ be a finite additive abelian group and $\mathcal{L}$ a finite set. A \textit{labeling function} is any function $L:\mathcal{G}\rightarrow \mathcal{L}$.
\end{definition}

In the group-based model with underlying finite abelian group $\mathcal{G}$, states are in bijection with the elements of the group $\mathcal{G}$. Fundamental in the definition of a group-based model associated to a finite additive abelian group $\G$ and a labeling function $L$ is that the rate of mutation from a state $g$ to state $h$ depends only on $L(h-g)$: That is, the entries of a rate matrix $Q$ are
$$
Q_{g,h}=\psi(h-g)
$$
for a vector $\psi \in \mathbb{R}^{\G}$ satisfying $\sum_{g \in \G} \psi(g)=0$, $\psi(g) \geq 0$ for all non-zero $g \in \G$ and $\psi(g)=\psi(h)$, whenever $L(g)=L(h)$. We say that such $Q$ is a $(\G,L)$-rate matrix. In this paper, the rate matrices in group-based models are assumed to be symmetric, i.e., $\psi(-g)=\psi(g)$ for every $g\in \mathcal{G}$. Since the matrix exponential of a symmetric matrix is again symmetric, then the entries of a transition matrix $P=\exp(Q)$ are
$$
P_{g,h}=f(h-g)
$$
for a vector $f \in \mathbb{R}^{\G}$ satisfying $\sum f(g)=1$, $f(g) \geq 0$ for all $g \in \G$ and $f(g)=f(-g)$ for all $g \in \G$.   As we see in Example~\ref{example:group_structure_not_preserved}, in general it is not true that $f(g)=f(h)$ whenever $L(g)=L(h)$. In Section~\ref{Section:labelling}, we introduce $\G$-compatible labeling functions which guarantee this property and then we say that $P$ is a $(\G,L)$-Markov matrix.

\begin{example} \label{example:labeling-functions-Z2xZ2}{\bf}{\rm}
Let $\mathcal{G}=\mathbb{Z}_2\times\mathbb{Z}_2$ and $\mathcal{L}=\{0,1,2,3\}$. The Kimura 3-parameter, the Kimura 2-parameter  and the Jukes-Cantor models correspond to the following labeling functions
\begin{align*}
L((0,0))=0, L((0,1))=1, L((1,0))=2, L((1,1))=3,\\
L((0,0))=0, L((0,1))=L((1,0))=1, L((1,1))=2,\\
L((0,0))=0, L((0,1))=L((1,0))=L((1,1))=1,
\end{align*}
respectively. The Kimura 3-parameter rate and transition matrices have the form
\begin{equation} \label{eqn:K3P-mutation-matrix}
\begin{pmatrix}
 		a&b&c&d\\b&a&d&c\\c&d&a&b\\d&c&b&a\\
 	\end{pmatrix}.
\end{equation}
In the case of the Kimura 2-parameter model additionally $b=c$ and in the case of the Jukes-Cantor model $b=c=d$.
 \end{example}

\begin{example} \label{example:group_structure_not_preserved}
Let $\G=\mathbb{Z}_7$ and $L$ be a labeling function such that $L(1)=L(2)=L(5)=L(6)$ and $L(3)=L(4)$. Consider the $(\G,L)$-rate matrix $$Q=\begin{pmatrix}
-1& 0.125& 0.125& 0.25& 0.25& 0.125& 0.125\\
  0.125& -1& 0.125& 0.125& 0.25& 0.25& 0.125\\
  0.125& 0.125& -1& 0.125& 0.125& 0.25& 0.25\\
  0.25& 0.125& 0.125& -1& 0.125& 0.125& 0.25\\
  0.25& 0.25& 0.125& 0.125& -1& 0.125& 0.125\\
  0.125& 0.25& 0.25& 0.125& 0.125& -1& 0.125\\
  0.125& 0.125& 0.25& 0.25& 0.125& 0.125&-1\\\end{pmatrix}.$$
In this rate matrix, $\psi(1)=\psi(2)=\psi(5)=\psi(6)=0.125$ and $\psi(3)=\psi(4)=0.25$. By direct computation, we get
    $$P=e^Q=
    \begin{pmatrix}
    0.41305& 0.0858551& 0.0834148& 0.124205& 0.124205& 0.0834148&0.0858551\\
  0.0858551& 0.41305& 0.0858551& 0.0834148& 0.124205&
  0.124205& 0.0834148\\
  0.0834148& 0.0858551& 0.41305& 0.0858551&
  0.0834148& 0.124205& 0.124205&\\
  0.124205& 0.0834148& 0.0858551&
  0.41305& 0.0858551& 0.0834148& 0.124205\\
  0.124205& 0.124205& 0.0834148& 0.0858551& 0.41305& 0.0858551& 0.0834148\\
  0.0834148&0.124205& 0.124205& 0.0834148& 0.0858551& 0.41305& 0.0858551&\\
  0.0858551& 0.0834148& 0.124205& 0.124205& 0.0834148&
  0.0858551& 0.41305\\
    \end{pmatrix}.$$
 The matrix $P$ is not a $(\G,L)$-Markov matrix, since $0.0858551=f(1)=f(6)\neq f(2)=f(5)=0.0834148$. The entries of $P$ induce a labeling function $L'$ such that $L'(1)=L'(6)\neq L'(2)=L'(5)$ and $L'(3)=L'(4).$ In this case, the matrix $P$ is a $(\G,L')$-Markov matrix.
\end{example}

Example~\ref{example:group_structure_not_preserved} shows that the matrix exponential does not necessarily preserve the labeling function associated to a rate matrix. Conversely,  Example 3.1 of \cite{RocaFernandez} suggests that a Kimura 3-parameter Markov matrix can be embeddable, despite the fact that it does not have any Markov generator satisfying Kimura 3-parameter model constraints.

Let $\mathbb{C}^*$ denote the multiplicative group of complex numbers without zero. A group homomorphism from $\G$ to $\mathbb{C}^*$ is called a \emph{character} of $\G$. The characters of $\G$ form a group under multiplication, called the \emph{character group} of $\G$ and denoted by $\widehat{\G}$. Here the product of two characters $\chi_1,\chi_2$ of the group $\mathcal{G}$ is defined by $(\chi_1\chi_2)(g)=\chi_1(g)\chi_2(g)$ for every $g\in\mathcal{G}$. The character group $\widehat{\G}$ is isomorphic to $\G$. Given a group isomorphism between $\G$ and $\widehat{\G}$, we will denote by $\widehat{g}\in \widehat{\G}$ the image of $g\in \G$.

\begin{lemma}[\cite{pachter2005algebraic}, Lemma 17.1] \label{lemma:properties-of-characters}
Let $g,h \in \G$ and $k \in \mathbb{Z}$. Then $\widehat{g}(-h)=\overline{\widehat{g}(h)}$ and $\widehat{kg}(h)=\widehat{g}(kh)$, where $\overline{a}$ denotes the complex conjugate of $a \in \C$.
\end{lemma}
It follows from Lemma~\ref{lemma:properties-of-characters} that the values of characters are roots of unity.

Given a function $a: \G \rightarrow \mathbb{C}$, its \emph{discrete Fourier transform} is a function $\check{a}: \G \rightarrow \mathbb{C}$ defined by
$$
\check{a}(g):=\sum_{h \in \G} \widehat{g}(h)a(h).
$$
\begin{lemma}[\cite{Matsen}, Section~2]\label{remark:real-valued-DFT}
For any real-valued function $a:\G \rightarrow \mathbb{C}$, the identity $\check{a}(-g)=\overline{\check{a}(g)}$ holds for all $g \in \G$. Moreover, if $a(-g)=a(g)$ for all $g \in \G$, then $\check{a}(-g)=\check{a}(g)$ for all $g \in \G$ and $\check{a}$ is a real-valued function.
\end{lemma}

In the proof of Theorem~\ref{theorem:embeddable_matrices}, we will use that $\check{\psi}$ and $\check{f}$ are real-valued. For this reason, in this paper we consider only group-based models that are equipped with \textit{symmetric} labeling functions, i.e. $L(g)=L(-g)$ for all $g\in \G$. In other words, a symmetric group-based model assumes that the transition matrices are real symmetric matrices.

The discrete Fourier transform is a linear endomorphism on $\mathbb{C}^{\G}$.  We will denote its matrix by $K$. In particular, the entries of $K$ are $K_{g,h}=\widehat{g}(h)$ for $g,h \in \G$. The matrix $K$ is symmetric for any finite abelian group~\cite[Section 3.2]{{luong2009fourier}}. The inverse of the discrete Fourier transformation matrix is $K^{-1}=\frac{1}{|\G|} K^*$, where $K^*$ denotes the adjoint of $K$~\cite[Corollary 3.2.2]{luong2009fourier}.

The following lemma describes the relation between functionals $f$ and $\psi$ in the case $f(-g)=f(g)$ and $\psi(-g)=\psi(g)$ for all $g \in \G$.

\begin{lemma}[\cite{Matsen}, Lemma 2.2] \label{lemma:exponential-of-DFT}
Let $Q$ be determined by $\psi \in \mathbb{R}^{\G}$ and $P$ be determined by $f \in \mathbb{R}^{\G}$ as described earlier in this section such that $P=e^Q$. Furthermore, assume that $\psi(g)=\psi(-g)$ and $f(g)=f(-g)$ for all $g \in \G$. Then, $\check{f}(g)=e^{\check{\psi}(g)}$ for all $g\in\mathcal{G}$.
\end{lemma}

\begin{lemma}\label{lemma:eigenpairs}
Let $Q$ be determined by $\psi \in \mathbb{R}^{\G}$ and $P$ be determined by $f \in \mathbb{R}^{\G}$ as described earlier in this section. Furthermore, assume that $\psi(g)=\psi(-g)$ and $f(g)=f(-g)$ for all $g \in \G$. Let $K_g$ denote the column of the discrete Fourier transform matrix labeled by $g$. The eigenpairs of $Q$ (resp.~$P$) are $(\check{\psi}(g),K_g)$ (resp.~$(\check{f}(g),K_g)$) for $g \in \mathcal{G}$.
\end{lemma}

\begin{proof}
    This result is stated in the proof of~\cite[Lemma 2.2]{Matsen}.
\end{proof}

In particular, in the case of a Markov matrix, the column vector of ones is an eigenvector with eigenvalue one. In the case of a rate matrix, the column vector of ones is an eigenvector with eigenvalue zero. A direct consequence of Lemma~\ref{lemma:eigenpairs} is that $Q$ and $P$ are diagonalizable by $K$, i.e. $Q=KD_1K^{-1}$ and $P=KD_2K^{-1}$ where $D_1$ and $D_2$ are diagonal matrices with diagonals given by the vectors $\check{\psi}$ and $\check{f}$ of $\mathbb{R}^{\G}$ respectively.

\section{\texorpdfstring{$\mathcal{G}$}{G}-compatible labeling functions}\label{Section:labelling}
In this section, we introduce a class of labeling functions with the property that the symmetries of the probability vector are preserved under the discrete Fourier transformation. This property is required for any result that is proven using the discrete Fourier transform. Notably, labeling functions for all common group-based models (CFN, K3P, K2P, and JC models) are $\G$-compatible.

\begin{definition}
Let $\G$ be a finite additive abelian group, $\mathcal{L}$ a set and $L:\mathcal{G}\rightarrow\mathcal{L}$ a labeling function. Let $K$ be the discrete Fourier transformation matrix for $\G$ and $x_L$ be the column vector of length $|\G|$ whose $g$-th component is the indeterminate $x_{L(g)}$. We say that $L$ is a \textit{$\mathcal{G}$-compatible labeling function} if  for every $g,h\in \mathcal{G}$ with $L(g)=L(h)$, we have that $K_{g,:} \cdot x_L=K_{h,:} \cdot x_L$ and $(K^{-1})_{g,:} \cdot x_L=(K^{-1})_{h,:} \cdot x_L$. Here $M_{a,:}$ denotes the row of $M$ indexed by group element $a.$
\end{definition}

A labeling function that maps every group element to a different label is trivially $\G$-compatible.

\begin{remark}\label{remark_compatible_labeling}
In the definition of a $\mathcal{G}$-compatible labeling, we require that the matrices $K$ and $K^{-1}$ preserve the symmetries of the vector of labels $x_L$. For symmetric group-based models, it is enough to require that only $K$ or $K^{-1}$ preserves the symmetries of the vector of labels $x_L$. Recall that
$$
K^{-1} \cdot x_L = \frac{1}{|\G|} \cdot K^* \cdot x_L = \frac{1}{|\G|} \cdot \overline K \cdot x_L.
$$
The property $\widehat{g}(-h)=\overline{\widehat{g}(h)}$ implies $K_{g,-h}=\overline{K_{g,h}}$ and $\overline{K_{g,-h}}=K_{g,h}$. If $-h=h$, this means $\overline{K_{g,h}}=K_{g,h}$ for all $g \in \G$. If $-h \neq h$, then taking into account that $x_{L(-h)}=x_{L(h)}$ gives $\overline{K_{g,h}} \cdot x_{L(h)}+\overline{K_{g,-h}} \cdot x_{L(-h)}=K_{g,h} \cdot x_{L(h)}+K_{g,-h} \cdot x_{L(-h)}$. Hence $K^{-1} \cdot x_L = 1/|\G|\cdot K \cdot x_L$.
\end{remark}

\begin{remark}
If a labeling function $L$ is symmetric $\G$-compatible and $Q$ is a $(\G,L)$-rate matrix, then a Markov matrix $P=e^Q$ is a $(\G,L)$-Markov matrix, i.e. $P_{g,h}=f(h-g)$ for a vector $f \in \mathbb{R}^G$ and $f(g)=f(h)$ whenever $L(g)=L(h)$.
\end{remark}

\begin{example} \label{example:labelings-of-Z2xZ2}
Let $\G= \mathbb{Z}_2\times\mathbb{Z}_2$. The discrete Fourier transformation matrix for $\G$ is
	$$K=\begin{pmatrix}
	1&1&1&1\\1&-1&1&-1\\1&1&-1&-1\\1&-1&-1&1\\
	\end{pmatrix}.$$
To show $\G$-compatibility for the three labeling functions from Example~\ref{example:labeling-functions-Z2xZ2}, it is enough to check that $K$ preserves the symmetries of $x_L$.
The labeling function of the Jukes-Cantor model is $\mathcal{G}$-compatible, since
		$$K\cdot \begin{pmatrix}
	x_0\\x_1\\x_1\\x_1\\
	\end{pmatrix}=\begin{pmatrix}
	x_0+3x_1\\x_0-x_1\\x_0-x_1\\x_0-x_1\\
	\end{pmatrix}.
	$$
The labeling function of the Kimura 2-parameter model is $\mathcal{G}$-compatible, since
$$
K\cdot
\begin{pmatrix}
x_0\\x_1\\x_1\\x_2\\
\end{pmatrix}=\begin{pmatrix}
x_0+2x_1+x_2\\x_0-x_2\\x_0-x_2\\x_0-2x_1+x_2\\
\end{pmatrix}.
$$
In the literature, usually $L((1,0))=L((1,1))$ in the Kimura 2-parameter model. The reason for this difference is that we use the discrete Fourier transform matrix in a format, which better demonstrates that it is the Kronecker product of discrete Fourier transformation matrices for $\mathbb{Z}_2$. The labeling function of the Kimura 3-parameter is $\mathcal{G}$-compatible, because each group element maps to a different label.
\end{example}

Sturmfels and Sullivant consider a different class of labeling functions, called friendly labelings~\cite{SturmfelsSullivant}.  Group-based models with friendly labeling functions are equivalent to $\mathcal{G}$-models defined by Micha\l ek~\cite[Remark 5.2]{michalek2011geometry}.  $\mathcal{G}$-models are constructed using an arbitrary group $\mathcal{G}$ that has a normal, abelian subgroup $\mathcal{H}$ which acts transitively and freely on the finite set of states. The importance of $\mathcal{G}$-models is that they are toric.  We explore connections between friendly labelings and $\mathcal{G}$-compatible labelings in Appendix~\ref{appendix:friendly-labelings}. We conjecture that every symmetric $\mathcal{G}$-compatible labeling is a friendly labeling, but not vice versa.

The following lemma provides a necessary condition for $\mathcal{G}$-compatible labeling functions.

\begin{lemma}\label{necessary_compatible_labeling}
Let $\G$ be a finite additive abelian group, $\mathcal{L}$ a set and $L:\mathcal{G}\rightarrow\mathcal{L}$ a labeling function. If $L$ is $\mathcal{G}$-compatible, then $L(0)\neq L(g)$ for any $g\neq 0$.
\end{lemma}

\begin{proof}
    Let $K$ be the discrete Fourier transformation matrix for $\G$. The entries of $K$ are $\widehat{g}(h)$ for $g,h \in \G$, which are roots of unity. The row $K_{0,:}$ consists of ones. On the other hand, no other row of $K$ consists of ones only, as this would contradict the uniqueness of the identity element in the character group. In particular, every other row of $K$ contains at least one element whose real part is strictly less than one. Thus it is impossible that $K_{0,:}\cdot x_L = K_{g,:}\cdot x_L$ for $g \neq 0$.
\end{proof}

Table~\ref{compatible_labeling} summarizes up to isomorphism all possible symmetric $\mathcal{G}$-compatible labeling functions for additive abelian groups of order up to eight. In the table, two group elements receive the same label if they belong to the same subset in a partition of $\mathcal{G}$.

    \begin{table}
     \centering
    \begin{tabular}{|p{0.75cm}|p{2cm}|p{14cm}|}
    \hline
        $n$ &  Group & Symmetric $\mathcal{G}$-compatible labelings \\
        \hline
         2 & $\mathbb{Z}_2$ & \{\{0\},\{1\}\} \\
         \hline
         3 & $\mathbb{Z}_3$ &
         \{\{0\},\{1,2\}\}\\
         \hline
         4 & $\mathbb{Z}_4$ &
         \{\{0\},\{1,2,3\}\},\{\{0\},\{1,3\},\{2\}\}\\
         \hline
         4 & $\mathbb{Z}_2\times \mathbb{Z}_2 $ &   \{\{(0,0)\},\{(0,1),(1,0),(1,1)\}\},\{\{(0,0)\},\{(0,1),(1,0)\},\{(1,1)\}\},\newline\{\{(0,0)\},\{(0,1)\},\{(1,0)\},\{(1,1)\}\}\\
         \hline
         5 & $\mathbb{Z}_5$ &
        \{\{0\},\{1,2,3,4\}\},
          \{\{0\},\{1,4\},\{2,3\}\}\\
         \hline
         6 & $\mathbb{Z}_2\times \mathbb{Z}_3 $ &\{\{(0,0)\},\{(0,1),(0,2),(1,0),(1,1),(1,2)\}\},\newline\{\{(0,0)\},\{(0,1),(0,2)\},\{(1,0)\},\{(1,1),(1,2)\}\}\\\hline
         7& $\mathbb{Z}_7 $ &\{\{0\},\{1,2,3,4,5,6\}\},\{\{0\},\{1,6\},\{2,5\},\{3,4\}\}\\\hline
         8& $\mathbb{Z}_8 $ & \{\{0\},\{1,2,3,4,5,6,7\}\},\{\{0\},\{1,3,5,7\},\{2,6\},\{4\}\},\newline\{\{0\},\{1,7\},\{2,6\},\{3,5\},\{4\}\}\\\hline
         8&$\mathbb{Z}_2\times\mathbb{Z}_4$ & \{\{(0,0)\},\{(0,1),(0,2),(0,3),(1,0),(1,1),(1,2),(1,3)\}\},\newline\{\{(0,0)\},\{(0,1),(0,3),(1,1),(1,3)\},\{(0,2)\},\{(1,0),(1,2)\}\},\newline\{\{(0,0)\},\{(0,1),(0,2),(0,3)\},\{(1,0)\},\{(1,1),(1,2),(1,3)\}\},\newline\{\{(0,0)\},\{(0,1),(0,3),(1,0)\},\{(0,2),(1,1),(1,3)\},\{(1,2)\}\},\newline\{\{(0,0)\},\{(0,1),(0,3)\},\{(0,2)\},\{(1,0)\},\{(1,1),(1,3)\},\{(1,2)\}\}\\\hline
         8&$\mathbb{Z}_2\times\mathbb{Z}_2\times\mathbb{Z}_2$ & \{\{(0,0,0)\},\{(0,0,1),(0,1,0),(0,1,1),(1,0,0),(1,0,1),(1,1,0),(1,1,1)\}\},
         \newline\{\{(0,0,0)\},\{(0,0,1)\},\{0,1,0),(1,0,0),(1,1,0)\},\{(0,1,1),(1,0,1),(1,1,1)\}\},
         \newline\{\{(0,0,0)\},\{(0,0,1),(1,0,0),(1,0,1)\}, \{(0,1,0)\},\{(0,1,1),(1,1,0),(1,1,1)\}\},
         \newline\{\{(0,0,0)\},\{(0,0,0),(0,1,0),(0,1,1)\},\{(1,0,0)\},\{(1,0,1),(1,1,0),(1,1,1)\}\},
         \newline\{\{(0,0,0)\},\{(0,0,1),(0,1,0),(1,0,1),(1,1,0)\},\{(0,1,1)\},\{(1,0,0),(1,1,1)\}\},
         \newline\{\{(0,0,0)\},\{(0,0,1),(0,1,0),(1,0,0)\},\{(0,1,1),(1,0,1),(1,1,0)\},\{(1,1,1)\}\},
         \newline\{\{(0,0,0)\},\{(0,0,1),(1,1,1)\},\{(0,1,0),(0,1,1),(1,0,0),(1,0,1)\},\{(1,1,0)\}\},
         \newline\{\{(0,0,0)\},\{(0,0,1),(0,1,1),(1,0,0),(1,1,0)\},\{(0,1,0),(1,1,1)\},\{(1,0,1)\}\},
         \newline\{\{(0,0,0)\},\{(0,0,1)\},\{(0,1,0),(1,0,0)\},\{(0,1,1),(1,0,1)\},\{(1,1,0)\},\{(1,1,1)\}\},
         \newline\{\{(0,0,0)\},\{(0,0,1),(0,1,0)\},\{(0,1,1)\},\{(1,0,0)\},\{(1,0,1),(1,1,0)\},\{(1,1,1\}\},
         \newline\{\{(0,0,0)\},\{(0,0,1),(1,0,0)\},\{(0,1,0)\},\{(0,1,1),(1,1,0)\},\{(1,0,1)\},\{(1,1,1\}\},
         \newline\{\{(0,0,0)\},\{(0,0,1)\},\{(0,1,0)\},\{(0,1,1)\},\{(1,0,0)\},\{(1,0,1)\},\{(1,1,0)\},\{(1,1,1)\}\}
         \\\hline
    \end{tabular}
    \caption{Symmetric $\mathcal{G}$-compatible labelings for abelian groups of order $n\leq8$ up to isomorphism.}
     \label{compatible_labeling}
    \end{table}

We saw in Example \ref{example:labelings-of-Z2xZ2} that the labeling function of the Jukes-Cantor model that assigns the same label to each nonzero element of the group $\G=\mathbb{Z}_2\times\mathbb{Z}_2$ is a $\G$-compatible labeling. This example can be generalized to other groups.

\begin{lemma}\label{lemma:general_jukes_cantor}
    Let $\mathcal{G}$ be a finite abelian group. Let $L:\mathcal{G}\rightarrow \{0,1\}$ be a labeling function such that and $L(0)=0$ and $L(g)=1$ for $g\neq 0$. Then the labeling function $L$ is symmetric $\mathcal{G}$-compatible.
\end{lemma}
\begin{proof}
Clearly the labeling function $L$ is symmetric. Let $K$ be the discrete  Fourier transformation matrix for $\G$. By~\cite[Corollary 3.2.1]{luong2009fourier}, we have
	\[
    \sum_{h\in \mathcal{G}}K_{g,h}= \left\{
    \begin{array}{ll}
    n,& g=0 \\
    0,& g\neq 0\\
    \end{array}
    \right..
    \]
Hence
$$
K_{g,:}\cdot x_L = \left\{
    \begin{array}{ll}
    x_0+(n-1)x_1,& g=0 \\
    x_0-x_1,& g\neq 0\\
    \end{array}
    \right..
$$
Hence $L$ is a $\mathcal{G}$-compatible labeling function.
\end{proof}

We call the model in Lemma~\ref{lemma:general_jukes_cantor} the \emph{general Jukes-Cantor model}. We finish this section with giving another class of labeling functions that are $\G$-compatible for every finite abelian group $\G$.

\begin{lemma}\label{strictly_symmetric_implies_compatible}
	Let $\G$ be a finite additive abelian group, $\mathcal{L}$ a finite set and $L:\mathcal{G}\rightarrow\mathcal{L}$ a labeling function such that for any two distinct elements $g,h\in \mathcal{G}$,  $L(g)=L(h)$ if and only if $g=-h$. Then $L$ is $\mathcal{G}$-compatible.
\end{lemma}
\begin{proof}
	 By Lemma~\ref{lemma:properties-of-characters}, the identity $\widehat{-g}(h)=\widehat{g}(-h)$ holds for all $g,h\in \mathcal{G}$. Then
	\begin{equation*}
	\begin{split}
	K_{-g,:}\cdot x_L&=\sum_{h\in\mathcal{G}}\widehat{-g}(h)x_{L(h)}=\sum_{h\in\mathcal{G}}\widehat{g}(-h)x_{L(h)}\\
	&=\sum_{h\in\mathcal{G}}\widehat{g}(h)x_{L(-h)}=\sum_{h\in\mathcal{G}}\widehat{g}(h)x_{L(h)}=	K_{g,:}\cdot x_L.
	\end{split}
	\end{equation*}
	Thus, the labeling function $L$ is $\mathcal{G}$-compatible as $x_L$ is the column vector of indeterminates $x_{L(g)}$.
\end{proof}
The converse of Lemma \ref{strictly_symmetric_implies_compatible} is not true in general. Two examples are given by the labeling functions for the Kimura 2-parameter and the Jukes-Cantor model.

\section{Model embeddability}
\label{Section:embeddability}

The following theorem is the main result of this paper. It characterizes $(\mathcal{G},L)$-embeddable transition matrices in terms of their eigenvalues.

\begin{theorem}\label{theorem:embeddable_matrices}
Fix a finite abelian group $\mathcal{G}$, a finite set $\mathcal{L}$, and a symmetric $\G$-compatible labeling function $L: \mathcal{G} \to \mathcal{L}$. Let $P$ be a $(\G,L)$-Markov matrix. Then $P$ is $(\G,L)$-embeddable if and only if the vector $\lambda \in \mathbb{R}^{\G}$ of eigenvalues of $P$ is in the set
\begin{align*}
\{\lambda \in \mathbb{R}^{\G}: & \lambda_0=1,  \prod_{h \in \mathcal{G}} \lambda_{h}^{\text{Re}((K)_{g,h})} \geq 1 \text{ for all nonzero } g\in\mathcal{G},\\
& \lambda_g>0 \text{ for all } g \in \G, \text{ and } \lambda_g=\lambda_h \text{ whenever } L(g)=L(h)\}.
\end{align*}
\end{theorem}

\begin{proof}

We start by summarizing the idea of the proof. We consider the set $\Psi_{\G,L}$ that consists of vectors $\psi$ that determine $(\G,L)$-rate matrices. Our goal is to characterize the set $\check{F}_{\G,L}$  of eigenspectra of Markov matrices that are matrix exponentials of $(\G,L)$-rate matrices  determined by vectors $\psi$ in $\Psi_{\G,L}$.
The first step is to consider the discrete Fourier transform of the set $\Psi_{\G,L}$, which we denote by $\check{\Psi}_{\G,L}$.
 By Lemma \ref{lemma:eigenpairs}, this set is the set of eigenvalues of the $(\G,L)$-rate matrices. The second step is to consider the image of the set $\check{\Psi}_{\G,L}$ under coordinatewise exponentiation. This set is precisely $\check{F}_{\G,L}$,
 because $(\G,L)$-rate matrices are diagonalizable by the discrete Fourier transform matrix $K$ by the discussion after Lemma~\ref{lemma:eigenpairs} and thus if a $(\G,L)$-rate matrix $Q$ is determined by $\psi \in \mathbb{R}^{\G}$ then
$$P=e^Q=K \cdot e^{\text{diag}(\check{\psi})} \cdot K^{-1}=K  \cdot \text{diag}(e^{\check{\psi}}) \cdot K^{-1},$$
where $\check{\psi}$ is the vector of eigenvalues of $Q$ and $e^{\check{\psi}}$ is the vector of eigenvalues of $P$.

More specifically, let
\begin{align*}
\Psi_{\G,L}=\{\psi \in \mathbb{R}^{\G}: &\sum_{g \in \G} \psi(g)=0, \psi(g) \geq 0 \text{ for all nonzero } g \in \G, \text{and}\\
& \psi(g)=\psi(h) \text{ whenever } L(g)=L(h)\}.
\end{align*}
The vectors in the set $\Psi_{\G,L}$ are in one-to-one correspondence with $(\G,L)$-rate matrices. The image of $\Psi_{\G,L}$ under the discrete Fourier transform is the set
\begin{align*}
\check{\Psi}_{\G,L}=
\{\check{\psi} \in \mathbb{R}^{\G}: &\check{\psi}(0)=0,  (K^{-1}\check{\psi})(g)\geq 0 \text{ for all nonzero } g\in\mathcal{G},
\text{ and}\\& \check{\psi}(g)=\check{\psi}(h) \text{ whenever } L(g)=L(h)\}.
\end{align*}
By Lemma \ref{lemma:eigenpairs}, this set is the set of eigenvalues of the $(\G,L)$-rate matrices.

The image of $\check{\Psi}_{\G,L}$ under the coordinatewise exponentiation is the set of eigenvalues of the $(\G,L)$-Markov matrices, which we denote by $\check{F}_{\G,L}$. We claim that $\check{F}_{\G,L}$ is equal to the set
\begin{equation} \label{eqn:F_check}
\begin{aligned}
\{\check{f} \in \mathbb{R}^{\G}: & \check{f}(0)=1,  \prod_{h\in \mathcal{G}}(\check{f}(h))^{(K^{-1})_{g,h}}\geq 1 \text{ for all nonzero } g\in\mathcal{G},\\
& \check{f}(g)>0 \text{ for all } g \in \G, \text{ and } \check{f}(g)=\check{f}(h) \text{ whenever } L(g)=L(h)\}.
\end{aligned}
\end{equation}
Indeed, let $\check{f}=\exp(\check{\psi})$.
Then $\check{f}>0$ because the image of the exponentiation map is positive. The inequality $a^Tx\geq 0$ is equivalent to $\exp(a^Tx)\geq1.$ Hence, the equation $\check{\psi}(0)=0$ gives $\check{f}(0)=1$ and the inequalities $(K^{-1}\check{\psi})(g)\geq 0$ give
\begin{align} \label{inequalities1}
\prod_{h\in \mathcal{G}}(\check{f}(h))^{(K^{-1})_{g,h}}=\prod_{h\in \mathcal{G}}(e^{(\check{\psi}(h))})^{(K^{-1})_{g,h}}=e^{\sum_{h\in\mathcal{G}}\check{\psi}(h)(K^{-1})_{g,h}} = e^{(K^{-1}\check{\psi})(g)} \geq 1
\end{align}
for all nonzero $g\in\mathcal{G}$. Hence $\check{f}$ is in the set~(\ref{eqn:F_check}). Conversely, let $\check{f}$ be a vector in the set~(\ref{eqn:F_check}).  Then $\log(\check{f}) \in \check{\Psi}_{\G,L}$ and $\check{f}=\exp(\log(\check{f}))$. Hence $\check{f}$ is in the image of $\check{\Psi}_{\G,L}$. Thus $\check{F}_{\G,L}$ is equal to the set~(\ref{eqn:F_check}).

It is left to rewrite the inequalities~(\ref{inequalities1}) as in the statement of the theorem.  We have
$$(K^{-1})_{g,-h}=\frac{1}{|\G|}\overline{K_{-h,g}}=\frac{1}{|\G|}\overline{\widehat{-h}(g)}=\frac{1}{|\G|}\overline{\widehat{h}(-g)}=\frac{1}{|\G|}\widehat{h}(g) =\frac{1}{|\G|}K_{h,g} = \overline{(K^{-1})_{g,h}}$$
for all $g,h\in\mathcal{G}$. Here we use Lemma~\ref{lemma:properties-of-characters} and the definition of the discrete Fourier transformation matrix. If $-h=h$, then $(K^{-1})_{g,h}=(K^{-1})_{g,-h}=\overline{(K^{-1})_{g,h}}$, and hence $(K^{-1})_{g,h}=\text{Re}((K^{-1})_{g,h})$.
If $-h \neq h$, then $\check{f}(h)=\check{f}(-h)$ by~Lemma~\ref{remark:real-valued-DFT}. Hence
\begin{equation}
\begin{split}
(\check{f}(h))^{(K^{-1})_{g,h}} (\check{f}(-h))^{(K^{-1})_{g,-h}} &=  (\check{f}(h))^{(K^{-1})_{g,h}} (\check{f}(h))^{\overline{(K^{-1})_{g,h}}}\\
&=(\check{f}(h))^{2\text{Re}((K^{-1})_{g,h})}\\
&=(\check{f}(h))^{\text{Re}((K^{-1})_{g,h})} (\check{f}(-h))^{\text{Re}((K^{-1})_{g,-h})}.
\end{split}
\end{equation}

We replace $K^{-1}$ by $1/|\G|\cdot \overline{K}$ and take both sides of the resulting inequality to the power $|\G|$. Finally, making the substitution $\lambda_{{h}}=\check{f}(h)$ gives the desired characterization.
\end{proof}

For $\G$ cyclic, Theorem~\ref{theorem:embeddable_matrices} has been independently proven by Baake and Sumner in the context of circulant matrices~\cite[Theorem 5.7]{Baake}. Moreover, they show that every embeddable circulant matrix is circulant embeddable~\cite[Corollary 5.2]{Baake}.

It follows from Lemma~\ref{lemma:exponential-of-DFT} that if a $(\G,L)$-Markov matrix $P$ is $(\G,L)$-embeddable, then there exists a unique $(\G,L)$-rate matrix $Q$ such that $P=\exp(Q)$. Indeed, since $Q$ and $P$ have both real eigenvalues and the eigenvalues of $P$ are exponentials of eigenvalues of $Q$, then the eigenvalues of $Q$ are uniquely determined by the eigenvalues of $P$. Then the $(\G,L)$-rate matrix $Q$ is the principal logarithm of $P$.

The inequalities $\lambda_g >0$ in Theorem~\ref{theorem:embeddable_matrices} imply $\det(P)=\prod\lambda_{g}>0$. Hence the set of $(\mathcal{G},L)$-embeddable matrices for a symmetric group-based model is a relatively closed subset of a connected component of the complement of $\det(P)=0$. A relatively closed subset means here a set that can be written as the intersection of a closed subset of $\mathbb{R}^{\mathcal{G} \times \mathcal{G}}$ and the connected component of the complement of $\det(P)=0$.

In the rest of the current section and in~Section~\ref{section:Hachimoji}, we will discuss applications of Theorem~\ref{theorem:embeddable_matrices}. We will recover known results about $(\G,L)$-embeddability and as a novel application characterize embeddability for three group-based models of hachimoji DNA.

\begin{example} The CFN model is the group-based model associated to the group $\mathbb{Z}_2$. The CFN Markov matrices have the form
$$
P=\begin{pmatrix}
a&b\\b&a\\
\end{pmatrix}.
$$
The discrete Fourier transform matrix is
$$
K=\begin{pmatrix}
1 & 1 \\
1 & -1
\end{pmatrix}.
$$
The eigenvalues of $P$ are $\lambda_0=a+b=1$ and $\lambda_1=a-b$. By Theorem~\ref{theorem:embeddable_matrices}, the Markov matrix $P$ is CFN embeddable if and only if $0 < \lambda_1 \leq 1$ or equivalently $0 < a-b \leq 1$. This is equivalent to $P$ satisfying $\det(P)>0$, or equivalently $\text{tr}(P)>1$. The result that a general $2\times 2$ stochastic matrix is embeddable if and only if $\det(P)>0$ or $\text{tr}(P)>1$ goes back to \cite[Proposition 2]{Kingman}. Hence $P$ is CFN embeddable if and only if it is embeddable.
\end{example}

\begin{example} Recall that the Kimura 3-parameter model is the group-based model associated to group $\mathcal{G}=\mathbb{Z}_2 \times \mathbb{Z}_2$ and a K3P Markov matrix $P$ has the form~(\ref{eqn:K3P-mutation-matrix}).
The eigenvalues of $P$ are
$$
\lambda_{(0,0)}=a+b+c+d,
\lambda_{(0,1)}=a-b+c-d,
\lambda_{(1,0)}=a+b-c-d,
\lambda_{(1,1)}=a-b-c+d.
$$
By Theorem~\ref{theorem:embeddable_matrices}, a Markov matrix $P$ is K3P embeddable if and only if
\begin{equation}\label{Z_2xZ_2_mutation_matrix_conditions}
\begin{aligned}
&\lambda_{(0,0)}=1,\lambda_{(0,1)}>0,\lambda_{(1,0)}>0,\lambda_{(1,1)}>0,\\
&\lambda_{(0,1)} \geq \lambda_{(1,0)}\lambda_{(1,1)}, \lambda_{(1,0)} \geq \lambda_{(0,1)}\lambda_{(1,1)}, \lambda_{(1,1)} \geq \lambda_{(0,1)}\lambda_{(1,0)}.
\end{aligned}
\end{equation}
This characterization for the Kimura 3-parameter model appears in~\cite[Theorem 3.2]{RocaFernandez}.

In the Kimura 2-parameter model $b=c$ and $\lambda_{(0,1)}=\lambda_{(1,0)}$. We get the conditions for the K2P embeddability by setting $\lambda_{(0,1)}=\lambda_{(1,0)}$ in~(\ref{Z_2xZ_2_mutation_matrix_conditions}). Hence a K2P Markov matrix is K2P embeddable if and only if
$$\lambda_{(0,0)}=1,\lambda_{(0,1)}>0,1 \geq \lambda_{(1,1)}\geq\lambda_{(0,1)}^2.$$
In the Jukes-Cantor model $b=c=d$ and $\lambda_{(0,1)}=\lambda_{(1,0)}=\lambda_{(1,1)}$. A JC Markov matrix is JC embeddable if and only if
$$\lambda_{(0,0)}=1,1 \geq \lambda_{(0,1)}>0.$$ These two characterizations are discussed in~\cite[Section 3]{RocaFernandez}.

The K3P embeddability of a K3P Markov matrix with no repeated eigenvalues is equivalent to the embeddability of the matrix. Similarly, the JC embeddability of a JC Markov matrix is equivalent to the embeddability of the matrix. The same is not true for K2P Markov matrices with exactly two coinciding eigenvalues. For further discussion see~\cite[Section 3]{RocaFernandez}.
\end{example}

\begin{remark} \rm
By~\cite[Corollary on page 18]{Kingman}, the map from rate matrices to transition matrices is locally homeomorphic except possibly when the rate matrix has a pair of eigenvalues differing by a non-zero multiple of $2 \pi i$. Since for symmetric group-based models rate matrices are real symmetric, then all their eigenvalues are real and hence the map from rate matrices to transition matrices is a homeomorphism. Therefore the boundaries of embeddable transition matrices of symmetric group-based models are images of the boundaries of the rate matrices. For general Markov model, the boundaries of embeddable transition matrices are characterized in~\cite[Propositions 5 and 6]{Kingman}.
\end{remark}

\begin{corollary}
A $(\mathcal{G},L)$-embeddable transition matrix lies on the boundary of the set of $(\mathcal{G},L)$-embeddable transition matrices for a symmetric group-based model if and only if it satisfies at least one of the inequalities in Theorem~\ref{theorem:embeddable_matrices} with equality.
\end{corollary}

\section{Hachimoji DNA}
\label{section:Hachimoji}

In this section, we suggest three group-based models for a genetic system with eight building blocks recently introduced by Hoshika et al~\cite{Hoshika}, and then characterize model embeddability for the proposed group-based models. The genetic system is called \textit{hachimoji DNA}. It has four synthetic nucleotides denoted \texttt{S}, \texttt{B}, \texttt{Z}, and \texttt{P} in addition to the standard nucleotides adenine (\texttt{A}), cytosine (\texttt{C}), guanine (\texttt{G}) and thymine (\texttt{T}). Detailed descriptions of the four additional nucleotides are given in~\cite{Hoshika}. If in the standard $4$-letter DNA, the purines are \texttt{A} and \texttt{G} and the pyrimidines are \texttt{C} and \texttt{T}, then in the hachimoji system, there are additionally purine analogs \texttt{P} and \texttt{B}, and pyrimidine analogs \texttt{Z} and \texttt{S}. The hydrogen bonds occur between the pairs \texttt{A-T, C-G, S-B} and \texttt{Z-P}.

This DNA genetic system with eight building blocks can reliably form matching base pairs and can be read and translated into RNA. It is mutable without damaging crystal structure which is required for molecular evolution. Hachimoji DNA has potential application	in bar-coding, retrievable information storage, and self-assembling nanostructures.
	
The underlying group we suggest for the hachimoji DNA is $\mathbb{Z}_2\times\mathbb{Z}_2\times\mathbb{Z}_2$, since when restricted to the standard $4$-letter DNA it gives the group $\mathbb{Z}_2 \times \mathbb{Z}_2$ that is the underlying group for the standard DNA models. We identify the nucleotides with the group elements of $\mathbb{Z}_2 \times \mathbb{Z}_2 \times \mathbb{Z}_2$ as follows:
\begin{eqnarray*}
&\texttt{A}=(0,0,0),\texttt{C}=(0,0,1),\texttt{T}=(0,1,0),\texttt{G}=(0,1,1),\\
&\texttt{P}=(1,0,0),\texttt{Z}=(1,0,1),\texttt{S}=(1,1,0),\texttt{B}=(1,1,1).
\end{eqnarray*}
The discrete Fourier transformation matrix of the group $\mathbb{Z}_2\times\mathbb{Z}_2\times\mathbb{Z}_2$ is
\begin{equation} \label{eqn:Z2xZ2xZ2-DFT-matrix}
K=\begin{pmatrix}
1 & 1 & 1 & 1 & 1 & 1 & 1 & 1\\
1 & -1 & 1 & -1 & 1 & -1 & 1 & -1\\
1 & 1 & -1 & -1 & 1 & 1 & -1 & -1\\
1 & -1 & -1 & 1 & 1 & -1 & -1 & 1\\
1 & 1 & 1 & 1 & -1 & -1 & -1 & -1\\
1 & -1 & 1 & -1 & -1 & 1 & -1 & 1\\
1 & 1 & -1 & -1 & -1 & -1 & 1 & 1\\
1 & -1 & -1 & 1 & -1 & 1 & 1 & -1
\end{pmatrix}.
\end{equation}

\subsection{Hachimoji 7-parameter model}\label{sec:hachimoji-7-parameter-model}

The first model we propose is the analogue of the Kimura 3-parameter model and we will call it the \emph{hachimoji 7-parameter (H7P) model}. In the hachimoji 7-parameter model, each element of the group $\mathbb{Z}_2\times\mathbb{Z}_2\times\mathbb{Z}_2$ maps to a distinct label. Thus the labeling function is trivially $(\G,L)$-compatible. The H7P rate and transition matrices have the form
\begin{equation}\label{K3P8}
    \begin{pmatrix}
    a&b&c&d&e&f&g&h\\
    b&a&d&c&f&e&h&g\\
    c&d&a&b&g&h&e&f\\
    d&c&b&a&h&g&f&e\\
    e&f&g&h& a&b&c&d\\
    f&e&h&g& b&a&d&c\\
    g&h&e&f&c&d&a&b\\
    h&g&f&e& d&c&b&a\\
    \end{pmatrix}.
\end{equation}

The eigenvalues of a H7P Markov matrix are $$\begin{pmatrix}1,\lambda_{(0,0,1)},\lambda_{(0,1,0)}, \lambda_{(0,1,1)},\lambda_{(1,0,0)},\lambda_{(1,0,1)}, \lambda_{(1,1,0)}, \lambda_{(1,1,1)}\end{pmatrix}^T=K \cdot \begin{pmatrix}a,b,c,d,e,f,g,h\end{pmatrix}^T.$$ By Theorem~\ref{theorem:embeddable_matrices}, such a matrix is H7P embeddable if and only if all eigenvalues are positive and satisfy
\begin{eqnarray*}
&\lambda_{(0,0,0)}=1,\\
&\lambda_{(0,1,0)}\lambda_{(1,0,0)} \lambda_{(1,1,0)}  \geq \lambda_{(0,0,1)} \lambda_{(0,1,1)} \lambda_{(1,0,1)} \lambda_{(1,1,1)},\\
&\lambda_{(0,0,1)}\lambda_{(1,0,0)} \lambda_{(1,0,1)}  \geq \lambda_{(0,1,0)} \lambda_{(0,1,1)} \lambda_{(1,1,0)} \lambda_{(1,1,1)} ,\\
&\lambda_{(0,1,1)}\lambda_{(1,0,0)} \lambda_{(1,1,1)}  \geq \lambda_{(0,0,1)} \lambda_{(0,1,0)} \lambda_{(1,0,1)} \lambda_{(1,1,0)},\\
&\lambda_{(0,0,1)}\lambda_{(0,1,0)} \lambda_{(0,1,1)}  \geq \lambda_{(1,0,0)} \lambda_{(1,0,1)} \lambda_{(1,1,0)} \lambda_{(1,1,1)},\\
&\lambda_{(0,1,0)}\lambda_{(1,0,1)} \lambda_{(1,1,1)}  \geq \lambda_{(0,0,1)} \lambda_{(0,1,1)} \lambda_{(1,0,0)} \lambda_{(1,1,0)},\\
&\lambda_{(0,0,1)}\lambda_{(1,1,0)} \lambda_{(1,1,1)}  \geq \lambda_{(0,1,0)} \lambda_{(0,1,1)} \lambda_{(1,0,0)} \lambda_{(1,0,1)},\\
&\lambda_{(0,1,1)}\lambda_{(1,0,1)} \lambda_{(1,1,0)}  \geq \lambda_{(0,0,1)} \lambda_{(0,1,0)} \lambda_{(1,0,0)} \lambda_{(1,1,1)}.
\end{eqnarray*}

\subsection{Hachimoji 3-parameter model} \label{sec:hachimoji-3-parameter-model}
The second model we suggest specializes to the Kimura 2-parameter model when restricted to the standard $4$-letter DNA. We will call it the \emph{hachimoji 3-parameter (H3P) model}. We recall that in the Kimura 2-parameter model there are three distinct parameters for the rates of mutation: One parameter for a state remaining unchanged, one parameter for transversion from a purine base to a pyrimidine base or vice versa, and one parameter for transition to the other purine or to the other pyrimidine. We say that two bases are of the same type if they are both standard or synthetic bases. In the hachimoji 3-parameter model, there are the following parameters:
		\begin{itemize}
			\item $a$: the probability of a state remaining unchanged.
			\item $b$: the probability of a transversion from a purine base to a pyrimidine base or vice versa.
			\item  $c$: the probability of a transition to another purine or pyrimidine base of the same type (same type transitions).
			\item $d$: the probability of a transition to another purine or pyrimidine base of different type (different type transitions).
		\end{itemize}
The H3P rate and transition matrices have the form
		\begin{equation} \label{K2P8}
		P=\begin{pmatrix}
		a&b&b&c&d&b&b&d\\
		b&a&c&b&b&d&d&b\\
		b&c&a&b&b&d&d&b\\
		c&b&b&a&d&b&b&d\\
		d&b&b&d&a&b&b&c\\
		b&d&d&b&b&a&c&b\\
		b&d&d&b&b&c&a&b\\
		d&b&b&d&c&b&b&a\\
		\end{pmatrix}.
		\end{equation}
The labeling function of this model corresponds to the partition
		$$\{\{(0,0,0)\},\{(0,0,1),(0,1,0),(1,0,1),(1,1,0)\},\{(0,1,1)\},\{(1,0,0),(1,1,1)\}\},$$
which is $(\G,L)$-compatible by Table~\ref{compatible_labeling}.

The eigenvalues of a H3P Markov matrix are
\begin{eqnarray*}
&w:=\lambda_{(0,0,0)}=a+4b+c+2d=1,x:=\lambda_{(0,1,1)}=a-4b+c+2d,\\ &y:=\lambda_{(1,0,0)}=\lambda_{(1,1,1)}=a+c-2d, z:=\lambda_{(0,0,1)}=\lambda_{(0,1,0)}=\lambda_{(1,0,1)}=\lambda_{(1,1,0)}=a-c.
\end{eqnarray*}
By Theorem~\ref{theorem:embeddable_matrices}, a H3P Markov matrix $P$ is H3P embeddable if and only if the eigenvalues of $P$ satisfy
\begin{equation}\label{inequality_hachimoji_3-parameter}
w=1, 1\geq x>0,y>0,z>0, x \geq y^2, xy^2 \geq z^4.
\end{equation}

\subsection{Hachimoji 1-parameter model}
 The third model we suggest is the analogue of the Jukes-Cantor model and we will refer to it as \textit{hachimoji 1-parameter (H1P) model}. It is the simplest group-based model associated to the group $\mathbb{Z}_2 \times \mathbb{Z}_2 \times \mathbb{Z}_2$ and it is described by only two distinct parameters for the rates of mutation. The two parameters are for a state remaining the same and a state mutating to any other state. The corresponding labeling function is $(\G,L)$-compatible by Lemma~\ref{lemma:general_jukes_cantor}. The H1P rate and transition matrices have the form
\begin{equation}\label{JC8-model}
		\begin{pmatrix}
		a&b&b&b&b&b&b&b\\
		b&a&b&b&b&b&b&b\\
		b&b&a&b&b&b&b&b\\
		b&b&b&a&b&b&b&b\\
		b&b&b&b&a&b&b&b\\
		b&b&b&b&b&a&b&b\\
		b&b&b&b&b&b&a&b\\
		b&b&b&b&b&b&b&a\\
		\end{pmatrix}.
\end{equation}
The eigenvalues of a  H1P Markov matrix are $w:=\lambda_{(0,0,0)}=1$ and $x:=\lambda_g=a-b$ for $g\neq0$.
By Theorem~\ref{theorem:embeddable_matrices}, such a matrix is H1P embeddable if and only if its eigenvalues satisfy
\begin{equation} \label{eqn:H1-embeddability}
w=1 \text{ and } 1 \geq x >0.
\end{equation}

\begin{remark} \label{remark:general-Jukes-Cantor-embeddability}
The same conditions as in~(\ref{eqn:H1-embeddability}) characterize model embeddability for the general Jukes-Cantor model as defined in Lemma~\ref{lemma:general_jukes_cantor}. This is also a special instance of a more general result~\cite[Corollary 4.7]{Baake} on equal-input embeddability. If the order of $\G$ is even, then the notion of general embeddability is equivalent to the notion of model embeddability for the general Jukes-Cantor models by~\cite[Theorem 4.6]{Baake}.
\end{remark}

\section{Volume}
\label{section:volume}

 In this section we compute the relative volumes of model embeddable Markov matrices within some meaningful subsets of Markov matrices by taking advantage of the characterisation of embeddability in terms of eigenvalues.
The aim of this section is to describe how large the different sets of matrices are compared to each other and provide intuition of how restrictive is the hypothesis
of homogeneous continuous-time models.

We will focus on the hachimoji models and the generalization of the Jukes-Cantor model. We will use the following notation:
\begin{enumerate}
    \item[(i)] $\Delta$ is the set of all Markov matrices in a model.
    \item[(ii)] $\Delta_+$ is the subset of matrices in $\Delta$ with only positive eigenvalues.
    \item[(iii)] $\Delta_{dd}$ is the subset of diagonally dominant matrices in $\Delta$, i.e. matrices in $\Delta$ such that in each row the diagonal entry is greater or equal than the sum of all other entries.
    \item[(iv)] $\Delta_{me}$ is the subset of model embeddable transition matrices in $\Delta$.
\end{enumerate}
Biologically, the subspace $\Delta_{dd}$ of diagonally dominant matrices consists of matrices with probability of not mutating at least as large as the probability of mutating.  If a diagonally dominant matrix is embeddable, it has an identifiable rate matrix \cite{Cuthbert1, Cuthbert}, namely a unique Markov generator, which is crucial for proving the consistency of many phylogenetic reconstruction methods, such as those based on maximum likelihood methods \cite{CasanellasPetrovicUhler, Chang}. What is more, the set of Markov matrices with positive eigenvalues $\Delta_{+}$ includes the multiplicative closure of the transition matrices in the continuous-time version of the model  \cite{sumner2012lie}. We have the inclusions
$\Delta_{me}\subseteq \Delta_{+}\subseteq \Delta$ and $\Delta_{dd}\subseteq \Delta_+$. The volumes of these spaces are given for the Kimura 3-parameter model in~\cite[Theorem 4.1]{RocaFernandez}, for the Kimura 2-parameter model in~\cite[Proposition 5.1]{casanellas2020embedding} and for the Jukes-Cantor model in~\cite[Section 4]{RocaFernandez}.

The subsets $\Delta,\Delta_+,\Delta_{dd},$ and $\Delta_{me}$ can be described using the parameterization in terms of the entries of the Markov matrix or in terms of their eigenvalues. We  parameterize the relevant subsets of Markov matrices in terms of the eigenvalues of the Markov matrices and compute the volumes using these parametrizations. If $\varphi$ denotes the bijection from the set of entries of a Markov matrix in a particular model to the set of its eigenvalues and the matrix $J(\varphi)$ denotes the Jacobian matrix of the map $\varphi$, then the volume of any subset in the parametrization using entries of a Markov matrix will be $|det(J(\varphi))|$ times the volume in the parameterization using eigenvalues. Since the determinant of this Jacobian is constant for each of the three models we consider, the relative volumes of the set of model embeddable Markov matrices will not depend on the parameterization chosen.

\begin{proposition}
For the hachimoji 7-parameter model, consider $\Delta$, $\Delta_+$, and $\Delta_{dd}$ as subsets of $\R^7$ parameterized by $\lambda_{(0,0,1)},\ldots,\lambda_{(1,1,1)}$, the eigenvalues of a H7P Markov matrix. Then: \begin{enumerate*}
    \item[(i)] $V(\Delta)=\frac{256}{315};$
    \item[(ii)] $V(\Delta_+)=\frac{5}{144};$
    \item[(iii)] $V(\Delta_{dd})=\frac{2}{315};$
\end{enumerate*}
\end{proposition}

\begin{proof}
The entries of a H7P Markov matrix (\ref{K3P8}) are determined by a vector $(a,b,c,d,e,f,g,h)$. The entries of this vector can be expressed in terms of the eigenvalues as
$$
\begin{pmatrix}a,b,c,d,e,f,g,h\end{pmatrix}^T=K^{-1} \begin{pmatrix}1,\lambda_{(0,0,1)},\lambda_{(0,1,0)}, \lambda_{(0,1,1)},\lambda_{(1,0,0)},\lambda_{(1,0,1)}, \lambda_{(1,1,0)}, \lambda_{(1,1,1)}\end{pmatrix}^T,
$$
where $K$ is the discrete Fourier transform matrix (\ref{eqn:Z2xZ2xZ2-DFT-matrix}).  In terms of the entries or the eigenvalues of a H7P Markov matrix, the relevant subsets in this model are given by:
\begin{align*}
    &\Delta =\{(a,b,c,d,e,f,g,h)\in \mathbb{R}^8:\qquad a+b+c+d+e+f+g+h=1, \qquad a,b,c,d,e,f,g,h\geq 0\},\\
    &\Delta_+=\{P\in \Delta : \qquad \lambda_{(0,0,1)},\lambda_{(0,1,0)}, \lambda_{(0,1,1)},\lambda_{(1,0,0)},\lambda_{(1,0,1)}, \lambda_{(1,1,0)}, \lambda_{(1,1,1)}>0\},\\
    &\Delta_{dd}=\{P\in \Delta :\qquad a\geq b+c+d+e+f+g+h\},\mbox{ and }\\
\end{align*}
$\Delta_{me}$ is given by one equation and seven inequalities presented in Section~\ref{sec:hachimoji-7-parameter-model}.
Expressing all conditions defining $\Delta$, $\Delta_+$, and $\Delta_{dd}$ in terms of the eigenvalues $\lambda_{(0,0,1)},\ldots,\lambda_{(1,1,1)}$ allows us to compute volumes of these sets using \texttt{Polymake}~\cite{polymake:2000}.
\end{proof}

We are not able to compute the volume of the subspace of the H7P embeddable Markov matrices exactly. Instead we estimate the volume using the hit-and-miss Monte Carlo integration method~\cite{hammersley2013monte} implemented in \texttt{Mathematica}. Table~\ref{table:estimated-volume-hachimoji-7} summarizes the volume for various number of sample points. Table~\ref{table:relative-volume-hachimjoi-7} gives relative volumes for the relevant sets.

\begin{table}[ht]
\centering
\caption{The estimated volume of the set of H7P embeddable matrices using the hit-and-miss Monte Carlo integration with $n$ sample points.}
\label{table:estimated-volume-hachimoji-7}
\begin{tabular}{|c|c|c|c|c|}
    \hline
    $n$ & $10^4$ & $10^5$ & $10^6$ & $10^7$\\
     \hline
   $V(\Delta_{me})$ & $0.0015$ & $0.00197$ & $0.001946$ & $0.0019678$\\
   \hline
   $V(\Delta_{me} \cap \Delta_{dd})$ & $0.0008$ & $0.00084$ & $0.00085$ & $0.0008271$\\
   \hline
\end{tabular}
\end{table}

\begin{table}[ht]
\centering
\caption{The relative volumes for the hachimoji 7-parameter model. The volumes of $\Delta_{me}$ and $\Delta_{me} \cap \Delta_+$ are estimated using Monte Carlo integration with $10^6$ sample points.}
\label{table:relative-volume-hachimjoi-7}
\begin{tabular}{|c|c|c|c|}
    \hline
     & $\Delta$& $\Delta_+$ &$\Delta_{dd}$ \\
     \hline
    $\frac{V(\cdot)}{V(\Delta)}$ & 1 &  $\frac{175}{4096}= 0.042724609375$ & $\frac{1}{128}=0.0078125 $ \\
   \hline
   $\frac{V(\Delta_{me} \cap \,\, \cdot \,)}{V(\cdot)}$ & $\approx0.00239$ &  $\approx0.056045$ & $\approx 0.13388$  \\
   \hline
\end{tabular}
\end{table}

\begin{proposition}
     For the hachimoji 3-parameter model, consider $\Delta$, $\Delta_+$, $\Delta_{dd}$, and $\Delta_{me}$ as subsets of $\R^3$ parameterized by $x,y,z$, the eigenvalues of a H3P Markov matrix. Then:
    \begin{enumerate*}
    \item[(i)] $V(\Delta)=\frac{4}{3}$;
    \item[(ii)] $V(\Delta_+)=\frac{7}{16}$;
    \item[(iii)] $V(\Delta_{dd})=\frac{1}{6}$;
     \item[(iv)] $V(\Delta_{me})=\frac{1}{3}$;
     \item[(v)] $V(\Delta_{me} \cap \Delta_{dd}) \approx 0.136733$.
    \end{enumerate*}
\end{proposition}
\begin{proof}
 The entries of a H3P Markov matrix as in~(\ref{K2P8}) can be expressed in terms of the eigenvalues as
$$
a=\frac{1+x+2y+4z}{8}, b=\frac{1-x}{8}, c=\frac{1+x+2y-4z}{8}, d=\frac{1+x-2y}{8}.
$$
Expressing all conditions defining $\Delta$, $\Delta_+$, $\Delta_{dd}$, and $\Delta_{me}$ in terms of $x,y,z$ allows us to use the \texttt{Integrate} command in \texttt{Mathematica} to compute the desired volumes. For $V(\Delta_{me} \cap \Delta_{dd})$ we used the numerical integration command \texttt{NIntegrate}.
\end{proof}

\begin{figure}
    \centering
      \includegraphics[width=0.6\textwidth]{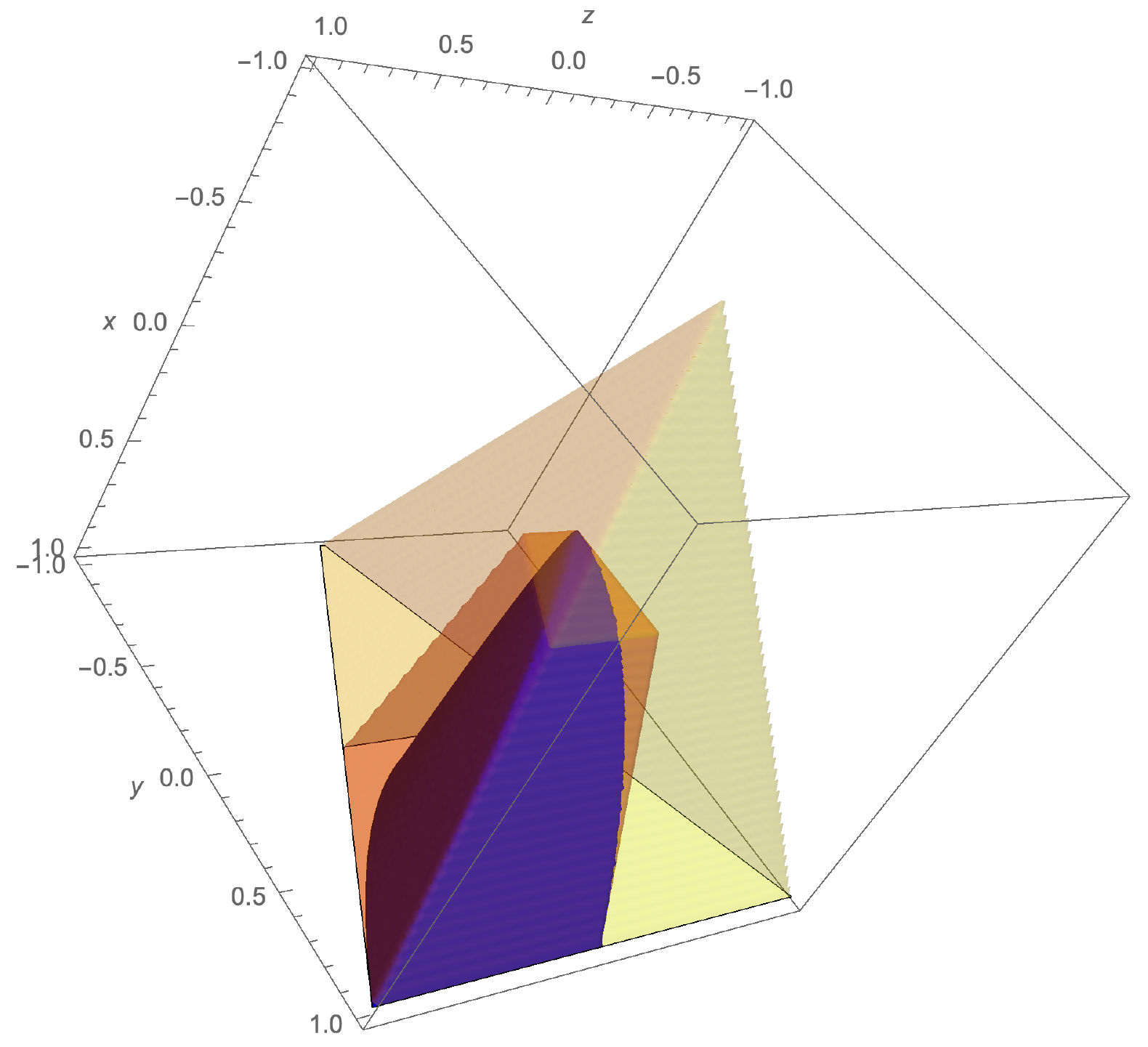}
    \caption{The sets $\Delta_{me}$, $\Delta_+$ and $\Delta$ for the hachimoji 3-parameter model. The sets $\Delta_+$ and $\Delta$ are polytopes; the set $\Delta_{me}$ is a semialgebraic set.}
    \label{fig:hachimoji3}
\end{figure}

The sets $\Delta_{me}$, $\Delta_+$ and $\Delta$ for the hachimoji 3-parameter model are depicted in Figure~\ref{fig:hachimoji3}. The relative volumes of relevant sets are given in Table~\ref{table:volumes-H3}.

\begin{table}[ht]
\centering
\caption{The relative volumes for the hachimoji 3-parameter model.}
\label{table:volumes-H3}
\begin{tabular}{|c|c|c|c|}
    \hline
     & $\Delta$& $\Delta_+$ &$\Delta_{dd}$ \\
     \hline
   $\frac{V(\cdot)}{V(\Delta)}$ & 1 & $\frac{21}{64} = 0.328125$ & $\frac{1}{8}=0.125$ \\
   \hline
   $\frac{V(\Delta_{me} \cap \,\, \cdot \,)}{V(\cdot)}$ & $\frac{1}{4}=0.25$ & $\frac{16}{21} \approx 0.76190$ & $\approx 0.82040$  \\
   \hline
\end{tabular}
\end{table}

Finally, we discuss the generalization of the Jukes-Cantor model which includes the hachimoji 1-parameter model. Let $\mathcal{G}$ be a finite abelian group of order $n$ and $L:\mathcal{G}\rightarrow \{0,1\}$ a labeling function such that and $L(0)=0$ and $L(g)=1$ for $g\neq 0$. In Lemma~\ref{lemma:general_jukes_cantor} we proved that $L$ is a $\mathcal{G}$-compatible labeling. In the general Jukes-Cantor model, the transition matrix   $P$ corresponding to this labeling is has the form
\[P_{ij}= \begin{cases}
      a, &  i=j \\
      b, & i\neq j
   \end{cases}
.\]
Since $P$ is a Markov matrix, then $a=1-(n-1)b$, and thus $P$ is parameterized by $b$.
\begin{proposition} \label{prop:volume-JC}
     For the general Jukes-Cantor model, consider $\Delta$, $\Delta_+$, $\Delta_{dd}$, $\Delta_{me}$ as subsets of $\R$ parameterized by $b$, the off-diagonal element of the Markov matrix. Then:
    \begin{enumerate*}
    \item[(i)] $\Delta=[0,\frac{1}{n-1}]$;
    \item[(ii)] $\Delta_+=[0,\frac{1}{n})$;
    \item[(iii)] $\Delta_{dd}=[0,\frac{1}{2(n-1)}]$;
     \item[(iv)] $\Delta_{me}=[0,\frac{1}{n})$;
     \item[(v)] $\Delta_{me} \cap \Delta_{dd}=[0,\frac{1}{2(n-1)}]$.
    \end{enumerate*}
\end{proposition}
\begin{proof}
The Markov matrix P has eigenvalues $1$ with multiplicity 1 and  $a-b=1-nb$ with multiplicity $n-1$. Hence

(i) $\Delta=\{b\in\mathbb{R}:a=1-(n-1)b \geq 0,b\geq 0\}= [0,\frac{1}{n-1}]$.

(ii) $\Delta_+=\{b\in\mathbb{R}:a=1-(n-1)b \geq 0,b\geq 0, 1-nb>0\}= [0,\frac{1}{n})$.

(iii) $\Delta_{dd}=\{b\in\mathbb{R}:a=1-(n-1)b \geq 0,b\geq 0, 1-(n-1)b\geq (n-1)b\}= [0,\frac{1}{2(n-1)}]$.

(iv) By Remark~\ref{remark:general-Jukes-Cantor-embeddability}, a Markov matrix is general Jukes-Cantor embeddable if and only if the eigenvalue $1-nb$ satisfies $1 \geq 1-nb >0$. Since $1 \geq 1-nb$ necessarily holds for any Markov matrix, we have $\Delta_{me}=\Delta_+$.

(v) Since $\Delta_{dd} \subseteq \Delta_+ = \Delta_{me}$, then $\Delta_{me} \cap \Delta_{dd}=\Delta_{dd}$.
\end{proof}

The relative volumes of relevant sets for the general Jukes-Cantor model are presented in Table~\ref{table:JC}.  Proposition~\ref{prop:volume-JC} gives for the hachimoji 1-parameter model \begin{enumerate*}
    \item[(i)] $\Delta=[0,\frac{1}{7}]$;
    \item[(ii)] $\Delta_+=[0,\frac{1}{8})$;
    \item[(iii)] $\Delta_{dd}=[0,\frac{1}{14}]$;
     \item[(iv)] $\Delta_{me}=[0,\frac{1}{8})$;
     \item[(v)] $\Delta_{me} \cap \Delta_{dd}=[0,\frac{1}{14}]$.
    \end{enumerate*}

\begin{table}[ht]
\centering
\caption{The relative volumes for the general Jukes-Cantor model.}
\label{table:JC}
\begin{tabular}{|c|c|c|c|}
    \hline
     & $\Delta$& $\Delta_+$ &$\Delta_{dd}$ \\
     \hline
     $\frac{V(\cdot)}{V(\Delta)}$ & 1& $\frac{n-1}{n}$ & $\frac{1}{2}$ \\
   \hline
   $\frac{V(\Delta_{me} \cap \,\, \cdot \,)}{V(\cdot)}$ & $\frac{n-1}{n}$&1&$1$ \\
   \hline
\end{tabular}
\end{table}

\section{Conclusion}
\label{section:Conclusion}

When modelling sequence evolution we often adopt several simplifying assumptions, which make the statistical problems tractable. The commonly used Markov models depend on the most common assumption that sites evolve independently following a Markov process. The Markov-chain is often assumed to be time-homogeneous, namely that substitution rates at any time are fixed and given by a rate matrix $Q$. Although in a globally time homogeneous process all branches have the same rate matrix, in the local time-homogeneous approach adopted in this paper we assume that each branch has a separate substitution rate. Violations of the local time-homogeneity is examined through non-embeddability of the models.

In this paper we provide necessary and sufficient conditions for model-embeddability of $n\times n$ symmetric group-based substitution models, which include the well known Cavender-Farris-Neyman, Jukes-Cantor, Kimura-2 and Kimura-3 parameter models for DNA. We fully characterize those embeddable  $n\times n$ stochastic matrices following a symmetric group-based model structure whose Markov generators also satisfy the constraints of the model, which we refer to as model embeddability.

A novel application of our main result is the characterization of model embeddability for three group-based models for the hachimoji DNA, a synthetic genetic system with eight building blocks. For these models we also compute the relevant volumes of model embeddable matrices within other relevant sets of Markov matrices. These computations show how restrictive is the hypothesis of a particular hachimoji time-continuous group-based model.

In this article we have considered symmetric group-based models. The importance of the symmetricity assumption is that it guarantees that the eigenvalues of rate and transition matrices of a group-based model are real. We use this property in the proof of Theorem~\ref{theorem:embeddable_matrices}. A future research question is to explore whether this approach can be extended to group-based models that are not symmetric.

\appendix
\section{Friendly labeling functions} \label{appendix:friendly-labelings}
Besides $\G$-compatible labeling functions, there is another class of labeling functions which has been studied in the literature. They are called \emph{friendly labeling functions} and were introduced by Sturmfels and Sullivant~\cite{SturmfelsSullivant}. Friendly labelings are useful in determining phylogenetic invariants for group-based models on evolutionary trees. In particular, a friendly labeling guarantees that if a particular labeling comes from an assignment of group elements, then any choice of a group element to one particular edge which is consistent with the labeling can be extended to an assignment that is consistent with labeling on all edges of the claw tree.

\begin{definition}
Let $\mathcal{G}$ be a finite abelian group and $L:\mathcal{G}\rightarrow\mathcal{L}$ a labeling function. Let $n\in \mathbb{N}$. Define $\tilde{L}:\mathcal{G}^n\rightarrow \mathcal{L}^n$ to be the induced labeling function on $\mathcal{G}^n$ and
$$Z=\{g\in \mathcal{G}^n: g_n=\sum_{i=1}^{n-1}g_i\}.$$
The labeling function $L$ is said to be $n$-\textit{friendly} if for every $l\in \tilde{L}(Z)$ and $i=1,2,\cdots,n$, we have
$\pi_i(\tilde{L}^{-1}(l))=L^{-1}(\pi_i(l))$.  Here, $\pi_i$ denotes the projection to the $i$-th component. Furthermore, the labeling function $L$ is said to be \textit{friendly} if it is $n$-friendly for all $n\geq 3.$
\end{definition}

By \cite[Lemma 11]{SturmfelsSullivant}, to check whether a labeling function is friendly, it is enough to check that the labeling is $3$-friendly.

\begin{example}[\cite{SturmfelsSullivant}, Example 9] \label{example_nonfriendly}
Let $\mathcal{G}=\mathbb{Z}_4$ and $L:\mathcal{G}\rightarrow\{0,1,2\}$ such that
$$L(0)=0, L(1)=1, L(2)=L(3)=2.$$
Then $L$ is not friendly labeling because $L^{-1}(\pi_3((1,1,2)))=\{2,3\}$ while $\pi_3(\tilde{L}^{-1}(1,1,2))=\pi_3((1,1,2))=\{2\}.$
\end{example}

Table~\ref{friendly_labeling} summarizes all friendly labelings for abelian groups of order $n$, where $2\leq n \leq 8$. In the table, two group elements receive the same label if they belong to the same subset in a partition of $\mathcal{G}$. In the table we do not include the friendly labelings for $\mathbb{Z}_2\times\mathbb{Z}_4$ and $\mathbb{Z}_2\times\mathbb{Z}_2\times\mathbb{Z}_2$, since there are too many of them.

    \begin{longtable}{|p{0.75cm}|p{2cm}|p{14cm}|}
    \hline
        n &  Group & Friendly labelings    \\
        \hline
         2 & $\mathbb{Z}_2$ & \{\{0,1\}\},\{\{0\},\{1\}\}\\
         \hline
         3 & $\mathbb{Z}_3$ &  \{\{0,1,2\}\},\{\{0\},\{1,2\}\},\{\{0\},\{1\},\{2\}\}\\
         \hline
         4 & $\mathbb{Z}_4$ &  \{\{0,1,2,3\}\},\{\{0\},\{1,2,3\}\},\{\{0,2\},\{1,3\}\},\{\{0\},\{1,3\},\{2\}\},\{\{0\},\{1\},\{2\},\{3\}\}\\
         \hline
         4&$\mathbb{Z}_2\times\mathbb{Z}_2$&\{\{(0,0),(0,1),(1,0),(1,1)\}\},\{\{(0,0)\},\{(0,1),(1,0),(1,1)\}\},
         \newline\{\{(0,0),(0,1)\},\{(1,0),(1,1)\}\},\{\{(0,0),(1,1)\},\{(0,1),(1,0)\}\},
         \newline\{\{(0,0),(1,0)\},\{(0,1),(1,1)\}\},\{\{(0,0)\},\{(0,1)\},\{(1,0),(1,1)\}\},
         \newline\{\{(0,0)\},\{(0,1),(1,0)\},\{(1,1)\}\},\{\{(0,0)\},\{(0,1),(1,1)\},\{(1,0)\}\},
         \newline\{\{(0,0)\},\{(0,1)\},\{(1,0)\},\{(1,1)\}\}\\\hline
         5 & $\mathbb{Z}_5$ &
        \{\{0,1,2,3,4\}\},\{\{0\},\{1,2,3,4\}\},\{\{0\},\{1,4\},\{2,3\}\},\{\{0\},\{1\},\{2\},\{3\},\{4\}\}\\
         \hline
         6&$\mathbb{Z}_2\times\mathbb{Z}_3$&
         \{\{(0,0),(0,1),(0,2),(1,0),(1,1),(1,2)\}\},\{\{(0,0)\},\{(0,1),(0,2),(1,0),(1,1),(1,2)\}\},
         \newline\{\{(0,0),(0,1),(0,2)\},\{(1,0),(1,1),(1,2)\}\},\{\{(0,0),(1,0)\},\{(0,1),(0,2),(1,1),(1,2)\}\},
         \newline\{\{(0,0)\},\{(0,1),(0,2)\},\{(1,0),(1,1),(1,2)\}\},\{\{(0,0)\},\{(0,1),(0,2),(1,1),(1,2)\},
         \newline\{(1,0)\}\},\{\{(0,0),(1,0)\},\{(0,1),(1,1)\},\{(0,2),(1,2)\}\},\newline\{\{(0,0)\},\{(0,1)\},\{(0,2)\},\{(1,0),(1,1),(1,2)\}\},
         \newline\{\{(0,0)\},\{(0,1),(0,2)\},\{(1,0)\},\{(1,1),(1,2)\}\},\{\{(0,0)\},\{(0,1),(1,1)\},\{(0,2),(1,2)\},\newline\{(1,0)\}\},
         \{\{(0,0)\},\{(0,1)\},\{(0,2)\},\{(1,0)\},\{(1,1)\},\{(1,2)\}\}\\\hline
         7 & $\mathbb{Z}_7$ &
          \{\{0,1,2,3,4,5,6\}\},\{\{0\},\{1,2,3,4,5,6\}\},\{\{0\},\{1,2,4\},\{3,5,6\}\},\newline\{\{0\},\{1,6\},\{2,5\},\{3,4\}\},\{\{0\},\{1\},\{2\},\{3\},\{4\},\{5\},\{6\}\}\\\hline
         8 & $\mathbb{Z}_8$&
         \{\{0,1,2,3,4,5,6,7\}\},\{\{0\},\{1,2,3,4,5,6,7\},\{\{0,4\},\{1,2,3,5,6,7\}\},
         \newline\{\{0,2,4,6\},\{1,3,5,7\}\},\{\{0\},\{1,2,3,5,6,7\},\{4\}\},\{\{0\},\{1,2,6,7\},\{3,4,5\}\},
         \newline\{\{0\},\{1,4,7\},\{2,3,5,6\}\},\{\{0\},\{1,3,5,7\},\{2,4,6\}\},\{\{0,4\},\{1,3,5,7\},\{2,6\}\},
         \newline\{\{0\},\{1,3,5,7\},\{2,6\},\{4\}\},\{\{0,4\},\{1,5\},\{2,6\},\{3,7\}\},\{\{0\},\{1,3,5,7\},\{2\},\{4\},\newline\{6\}\},\{\{0\},\{1,3\}\{2,6\},\{4\},\{5,7\}\},\{\{0\},\{1,7\},\{2,6\},\{3,5\},\{4\}\},\{\{0\},\{1,5\},\{2,6\},\newline\{3,7\},\{4\}\},\{\{0\},\{1,5\},\{2\},\{3,7\},\{4\},\{6\}\},\{\{0\},\{1\},\{2\},\{3\},\{4\},\{5\},\{6\},\{7\}\}\\\hline
             \caption{Friendly labelings for abelian group of order $n\leq 8$}
    \label{friendly_labeling}
    \end{longtable}

The next example shows that there are friendly labelings that are not $\G$-compatible.

\begin{example}
Let $\G$ be a finite abelian group and $L:\G\rightarrow\{0\}$ the labeling function defined by $L(g)=0$ for all $g \in \G$.  By Lemma~\ref{necessary_compatible_labeling}, the labeling function $L$ is not $\mathcal{G}$-compatible. However, it is a friendly labeling, since $\pi_i(\tilde{L}^{-1}((0,0,0)))=\pi_i(Z)=\G$ and  $L^{-1}(\pi_i((0,0,0)))=L^{-1}(0)=\G$.
\end{example}

Computations for abelian groups of order at most eight demonstrate that every symmetric $\G$-compatible labeling is a friendly labeling. It is left open, if the same is true for any finite abelian group.

\begin{question}
Given a finite abelian group $\mathcal{G}$, is the set of all symmetric $\mathcal{G}$-compatible labelings strictly contained in the set of all friendly labelings?
\end{question}

\smallskip

\smallskip

\noindent \textbf{Acknowledgements.}
An initial version of the main result of the present manuscript appeared in the preprint arXiv:1705.09228. Following the advice of referees, we divided the preprint into two manuscripts. The present manuscript builds on the section on the embedding problem in the earlier preprint. Most of the preprint arXiv:1705.09228 appeared in Bulletin of Mathematical Biology under the title ``Maximum Likelihood Estimation of Symmetric Group-Based Models via Numerical Algebraic Geometry''. Dimitra Kosta was partially supported by a Royal Society Dorothy Hodgkin Research Fellowship. Kaie Kubjas was partially supported by the Academy of Finland Grant 323416.

\bibliographystyle{plain}
\bibliography{main}

\smallskip

\smallskip

\noindent Authors' affiliations:

\vspace{0.2cm}
\noindent Muhammad Ardiyansyah, Department of Mathematics and Systems Analysis, Aalto University, \texttt{muhammad.ardiyansyah@aalto.fi}

\vspace{0.2cm}
\noindent Dimitra Kosta,
School of Mathematics, University of Edinburgh,\\
\texttt{D.Kosta@ed.ac.uk}

\vspace{0.2cm}
\noindent Kaie Kubjas, Department of Mathematics and Systems Analysis, Aalto University,\\ \texttt{kaie.kubjas@aalto.fi}

\end{document}